\documentclass[12pt]{iopart}

\usepackage{iopams}
\usepackage{amsthm,color} 
\usepackage{graphicx} 
\newtheorem{theorem}{Theorem}[section] 
\newtheorem{corollary}[theorem]{Corollary}
\newtheorem{proposition}[theorem]{Proposition}
\newtheorem{definition}[theorem]{Definition}

\begin{document}

\title{A note on the Relationship between Localization and Norm-1 Property}

\author{R. Beneduci}

\address{Department of Mathematics University of Calabria, Italy
\\ and\\
 Istituto Nazionale di Fisica Nucleare gruppo collegato Cosenza}
\ead{rbeneduci@unical.it}
\author{F. E. Schroeck Jr.}

\address{Department of Mathematics, University of Denver, Colorado (USA)
\\and\\
 Florida Atlantic State University (USA)}
\ead{fschroec@du.edu}
\begin{abstract}
\noindent
The paper focuses on the problem of localization in quantum mechanics. It is well known that it is not possible to define a localization observable for the photon by means of projection valued measures. Conversely, that is possible by using positive operator valued measures. On the other hand, projection valued measures imply a kind of localization which is stronger than the one implied by positive operator valued measure. It has been claimed that the norm-1 property would in some sense reduce the gap between the two kind of localizations. We give a necessary condition for the norm-1 property and show that it is not satisfied by several important localization observables. 

\end{abstract}

\maketitle

\section{Introduction}
In the standard formulation of quantum mechanics the observables of a quantum system are represented by self-adjoint operators. The spectral theorem \cite{Berberian} assures the existence of a one-to-one correspondence between self-adjoint operators and Projection Valued Measures (PVMs). In particular, for each self-adjoint operator $A$ there is a map (a PVM) $E:\mathcal{B}(\mathbb{R})\to\mathcal{E(H)}$ from the Borel $\sigma$-algebra of the reals to the space of projection operators on the Hilbert space $\mathcal{H}$ such that $E(\mathbb{R})=\mathbf{1}$, and, for each sequence of disjoint Borel sets $\{\Delta_i\}_{i\in\mathbb{N}}$, $\cup_{i=1}^{\infty}\Delta_i=\Delta$,
$$\sum_{i=1}^{\infty}E(\Delta_i)=E(\Delta)$$

\noindent
where the convergence is in the weak operator topology. 

As it is well known \cite{Jauch,Ali,Holevo1,Holevo2,Busch,Busch1,Rosewarne,Muga}, there are quantum observables (e.g. time observable, position observable for the photon, phase observable) that are not representable by means of self-adjoint operators or Projection Valued Measures.  

A fruitful way to overcome the problem is to generalize the concept of observable by means of Positive Operator Valued Measures (POVMs) \cite{Ali,Busch,Holevo1,Holevo2,Schroeck1} of which the PVMs are a special case. In the following $\mathcal{F(H)}$ denotes the space of positive linear operators less than or equal to the identity $\mathbf{1}$.
\begin{definition}
Let $X$ be a topological space and $\mathcal{B}(X)$ the Borel $\sigma$-algebra on $X$.
\label{POV}
A POVM is a map $F:\mathcal{B}(X)\to\mathcal{
F}(\mathcal{H})$
such that:
    \begin{equation*}
    F\big(\bigcup_{n=1}^{\infty}\Delta_n\big)=\sum_{n=1}^{\infty}F(\Delta_n).
    \end{equation*}
    \noindent 
 where, $\{\Delta_n\}$ is a countable family of disjoint
    sets in $\mathcal{B}(X)$ and the series converges in the weak operator topology. It is said to be normalized if 
\begin{equation*}   
    F(X)={\bf{1}}.
\end{equation*}
\end{definition}    
\begin{definition}
    A POVM is said to be commutative if
    \begin{equation}
    \big[F(\Delta_1),F(\Delta_2)\big]={\bf{0}}\,\,\,\,\forall\,\Delta_1\,,\Delta_2\in\mathcal{B}(X).
    \end{equation}
    \end{definition}

   \begin{definition}
   A POVM is said to be orthogonal if
    \begin{equation}
    F(\Delta_1)F(\Delta_2)={\bf{0}}\,\,\,\hbox{if}\,\,\Delta_1\cap\Delta_2=
    \emptyset.
    \end{equation}
\end{definition}
\begin{definition}
A PVM is an orthogonal, normalized POVM.
\end{definition}
\noindent
In the case of a PVM $E$ we have $\mathbf{0}=E(\Delta)[1-E(\Delta)]=E(\Delta)-E^2(\Delta)$. Therefore, $E(\Delta)$ is a projection operator for every $\Delta\in\mathcal{B}(X)$. We have proved the following proposition.
\begin{proposition} A PVM $E$ on $X$ is a map $E:\mathcal{B}(X)\to\mathcal{E}(\mathcal{H})$ from the Borel $\sigma$-algebra of $\mathcal{B}(X)$ to the space of projection operators on $\mathcal{H}$.
\end{proposition}
\begin{definition}
The spectrum $\sigma(F)$ of $F$ is the set of points $x\in X$ such that $F(\Delta)\neq\mathbf{0}$, for any open set $\Delta$ containing $x$.
\end{definition}
\noindent
From a general theoretical viewpoint, the introduction of POVMs can be justified by analyzing the statistical description of a measurement \cite{Holevo1} but, as we pointed out above, there are important physical motivations that go in the same direction.  (See also \cite{Busch,Schroeck1}.) 

The present note focuses on the problem of localization in quantum mechanics. We start by giving the definition of covariance which plays a key role in the definition of localization.
 \begin{definition}
Let $G$ be a locally compact topological group. Let $x\mapsto gx$, $g\in G$, be the action of $G$ on a topological space $X$. Let $U$ be a strongly continuous unitary representation of $G$ in Hilbert space $\mathcal{H}$. A POVM $F:\mathcal{B}(X)\to\mathcal{F(H)}$ is covariant with respect to $G$ if, for any Borel set $\Delta\in \mathcal{B}(X)$,
$$U_g^{\dag}F(\Delta)U_g=F(g\Delta)$$
\noindent
where, $g\Delta=\{x'\in X\,\,\vert\,\, x'=gx, x\in\Delta\}$.
\end{definition}
\noindent
Now, we can proceed to give the definition of localization. Localization requires covariance with respect to a group $G$ describing the kinematics of the system.We start by considering the general case of localization in an abstract topological space $X$. Then, we specialize $X$ to $\mathbb{R}^3$ (for the case of space localization) and  to the phase space $\Gamma$. (For the case of phase space localization.) We have two possible definitions of localization; i.e., sharp localization and unsharp localization. Sharp localization is defined by requiring covariance of a PVM under the group $G$.
\begin{definition}
Let $G$ be a group describing the kinematics of a quantum system. Let $x\mapsto gx$, $g\in G$, be the action of $G$ on a topological space $X$. Let $U$ be a strongly continuous unitary representation of $G$ in Hilbert space $\mathcal{H}$. A PVM $E:\mathcal{B}(X)\to\mathcal{F(H)}$ represents a sharp localization observable in $X$ (with respect to $(G,U)$) if 
$$U_gE(\Delta)U^\dag_g=E(g\Delta).$$
\end{definition}  
\noindent
The previous definition can be weakened by replacing the PVM by a POVM. That corresponds to an unsharp localization.  
\begin{definition}
Let $G$ be a group describing the kinematics of a quantum system. Let $x\mapsto gx$, $g\in G$, be the action of $G$ on a topological space $X$. Let $U$ be a strongly continuous unitary representation of $G$ in Hilbert space $\mathcal{H}$. A POVM $F:\mathcal{B}(X)\to\mathcal{F(H)}$ represents an unsharp localization observable in $X$ (with respect to $(U,G)$) if
 $$U_gE(\Delta)U^\dag_g=E(g\Delta).$$
\end{definition} 
\noindent
Now, we specialize to the case of relativistic localization in $\mathbb{R}^3$. In the relativistic case, the relevant group is the Poincare group and sharp localization is defined as follows \cite{Castrigiano,Holevo2}: let $W$ be a continuous unitary representation of the universal covering of the Poincare group. Let U
be the restriction of $W$ to the universal covering group $ISU(2)=\{(\mathbf{a},B),\,\vert\,\mathbf{a}\in\mathbb{R}^3,\,B\in SU(2)\}$ of the Euclidean group and $\Lambda:SU(2)\to SO(3)$ the universal covering homomorphism. 
 A quantum system is said to be Wightman localizable if there is a PVM $E:\mathcal{B}(\mathbb{R}^3)\to\mathcal{E(H)}$ such that
$$U(\mathbf{a},B)E(\Delta)U^\dag(\mathbf{a},B)=E(\mathbf{a}+\Lambda(B)\Delta) \quad (\rm{sharp\, localization}).$$
\noindent
The covariance assures that the results of a localization measurement do not depend on the choice of the origin and the orientation of the reference frame.
 \noindent
As we have just remarked, in the case of the photon, sharp localization is impossible \cite{Ali,Castrigiano,Jauch,Hegerfeldt}. Conversely, localization of the photon can be described by means of POVMs $F:\mathcal{B}(\mathbb{R}^3)\to\mathcal{F(H)}$ such that  
$$U(\mathbf{a},B)F(\Delta)U^\dag(\mathbf{a},B)=F(\mathbf{a}+\Lambda(B)\Delta) \quad (\rm{unsharp\, localization}).$$ 


By specializing $X$ to a phase space $\Gamma$ we get the concept of  unsharp localization in phase space which requires the POVM, $F$, to be defined on $\Gamma=X$ and to be covariant with respect to a group $G$ which characterizes the symmetries of the system \cite{Brooke}. Examples of symmetry groups are the Galilei group in the non relativistic case and the Poincare group in the relativistic case. (See section 3 for further details.) For the case of the photons in their phase space the relevant group is the Poincare group. \cite{Brooke}

Clearly the introduction of the POVMs for the description of localization observables implies a change in the standard concept of localization in quantum mechanics. If a localization observable is described by a covariant PVM $E$ (sharp localization) then, for any Borel set $\Delta$ there exists a state $\psi$ such that $\langle\psi, E(\Delta)\psi\rangle=1$; i.e., the probability that a measure of the position of the system in the state $\vert\psi\rangle$ gives a result in $\Delta$  is one. Conversely,  if a localization observable is described by a POVM $F$, there are Borel sets $\Delta$ such that  $0<\langle\psi, F(\Delta)\psi\rangle<1$ for any vector $\psi$ (unsharp localization). There are even covariant POVMs such that the condition $\mathbf{0}\leq F(\Delta)<\mathbf{1}$ holds for any $\Delta\in\mathcal{B}(X)$ with $\mu(X-\Delta)>0$, where $\mu$ is the measure on $X$. \cite{Schroeck} 

Sharp localizability being untenable in relativistic theory, \cite{Wightman} we need to switch to unsharp localizability. It is worth remarking that the relationships between localization and relativistic causality is quite problematic \cite{Hegerfeldt,Busch1}.


What is said above raises the following question for the unsharp case: Is it true that, for any Borel set $\Delta$, there is a family of unit vectors $\psi_n$ such that $\lim_{n\to\infty}\langle\psi_n, F(\Delta)\psi_n\rangle=1$. Whenever such a property holds, we say that the POVM has the norm-1 property. \cite{Car,He} Clearly, the norm-1 property implies a kind of unsharp localization which is closer to the sharp one. Indeed, for each $\epsilon$, it is possible to find a state $\psi$ such that the probability $\langle\psi, F(\Delta)\psi\rangle$ that a measure gives a result in $\Delta$ is greater than $1-\epsilon$.

In the present note we analyze some general aspects of the concept of the norm-1 property which are related to the concepts of absolute continuity and uniform continuity of a POVM (section 2) and derive some consequences for the concept of localization in phase space and configuration space (sections 3 and 4). In particular, we give a necessary condition for the norm-1 property to hold and prove that it is not satisfied for a class of localization observables in  phase space as well as for the corresponding marginal observables. 
\section{A necessary condition for the norm-1 property}

In the present section, we give a necessary condition for the norm-1 property of a POVM defined on a topological manifold of dimension $n$.
First, we recall the definition of the norm-1 property.
\begin{definition} A POVM $F:\mathcal{B}(X)\to\mathcal{F(H)}$ has the norm-1-property if $\|F(\Delta)\|=1$, for each $\Delta\in\mathcal{B}(X)$ such that $F(\Delta)\neq\mathbf{0}$. 
\end{definition}
\noindent
The following proposition points out the physical meaning of the norm-1 property. 
\begin{proposition}
$F$ has the norm-1 property if and only if, for each $\Delta\in\mathcal{B}(X)$ such that $F(\Delta)\neq\mathbf{0}$, there is a sequence of unit vectors $\psi_n$ such that $\lim_{n\to\infty}\langle\psi_n, F(\Delta)\psi_n\rangle=1$.
\end{proposition}
\begin{proof}
Suppose $F$ has the norm-1 property and $\|F(\Delta)\|\neq\mathbf{0}$. Then, 
$$1=\|F(\Delta)\|=\sup_{\|\psi\|=1}\{\langle\psi, F(\Delta)\psi\rangle\}.$$
\noindent
Hence, there is a sequence $\vert\psi_n\rangle$ such that 
\begin{equation}\label{lim}
\lim_{n\to\infty}\langle\psi_n, F(\Delta)\psi_n\rangle=1.
\end{equation}
\noindent
Conversely, since $F(\Delta)\leq\mathbf{1}$, equation (\ref{lim}) implies $\|F(\Delta)\|=1$.
\end{proof}
\noindent
In quantum mechanics, the state of a system is a unit vector, $\psi$, in a Hilbert space $\mathcal{H}$ while an observable is a PVM or a POVM, $F$.  From an operational viewpoint, the states represent the \textit{preparation instruments} while the observables represent the \textit{measurement instruments}.\cite{Ludwig,Muynck,Schroeck1} The connection between the two mathematical terms (states and observables) and the experimental data is given by the expression $p^F_\psi(\Delta):=\langle\psi, F(\Delta)\psi\rangle$ which is interpreted as the probability that the pointer of the measurement instrument (represented by $F$) gives a result in $\Delta$ if the state of the system is $\psi$. It is then clear why proposition \ref{lim} explains the physical meaning of the norm-1 property; i.e., if $F$ has the norm-1 property then, for any $\epsilon>0$, there is a state $\psi$ such that $p^F_\psi(\Delta)>1-\epsilon $. That implies a kind of localization very close to the one we can realize with the PVMs.

Now, we need to introduce the concept of absolute continuity which will be helpful in the study of localization in phase space. 
\begin{definition}
Let $F:\mathcal{B}(X)\to\mathcal{F(H)}$ be a POVM and $\nu:X\to\mathbb{R}$ a regular measure. $F$ is absolutely continuous with respect to $\nu$ if there exists a number $c$ such that 
\begin{equation*}
\|F(\Delta)\|\leq c\,\nu(\Delta),\quad \forall\Delta\in\mathcal{B}(X).
\end{equation*}
\end{definition}
\begin{definition}
Let $F$ be a POVM.  Let $\Delta=\cup_{i=1}^{\infty}\Delta_i$,  $\Delta_i\cap\Delta_j=\emptyset$. $F$ is said uniformly continuous if $\lim_{n\to\infty}\sum_{i=1}^{n}F(\Delta_i)=F(\Delta)$ in the uniform operator topology.
\end{definition}
\begin{proposition}\label{up}
A POVM $F$ is uniformly continuous if and only if, for any sequence $\{\Delta_i\}_{i\in\mathbb{N}}$such that $\Delta_i\uparrow\Delta$, we have  
\begin{equation*}
\lim_{n\to\infty}\| F(\Delta)-F(\Delta_i)\|=0.
\end{equation*}   
\end{proposition}
\begin{proof}
Suppose that $\lim_{n\to\infty}\| F(\Delta)-F(\Delta_i)\|=0$ whenever $\Delta_i\uparrow\Delta$. Let $\{\Delta_i\}_{i\in\mathbb{N}}$ be a sequence of disjoint sets such that $\cup_{i=1}^{\infty}\Delta_i=\Delta$. Then, we can define the family of sets $\overline{\Delta}_n=\cup_{i=1}^n\Delta_i$. We have $\overline{\Delta}_i\uparrow\Delta$. Therefore,
$$\lim_{n\to\infty}\|F(\Delta)-\sum_{i=1}^nF(\Delta_i)\|=\lim_{n\to\infty}\|F(\Delta)-F(\overline{\Delta}_n)\|=0$$
\noindent
Conversely, suppose that $F$ is uniformly continuous. Let $\Delta_i$ be such that $\Delta_i\uparrow\Delta$. We can define the family of sets $\overline{\Delta}_i=\Delta_i-\Delta_{i-1}$ with $\Delta_0=\emptyset$. We have, $\overline{\Delta}_i\cap\overline{\Delta}_j=\emptyset$, $i\neq j$. Moreover, $\Delta_n=\cup_{i=1}^n\overline{\Delta}_i$ and $\cup_{i=1}^{\infty}\overline{\Delta}_i=\Delta$. Therefore,
$$\lim_{n\to\infty}\| F(\Delta)-F(\Delta_n)\|=\lim_{n\to\infty}\|F(\Delta)-\sum_{i=1}^n F(\overline{\Delta}_i)\|=0.$$
\end{proof}
\begin{proposition}\label{down}
$F$ is uniformly continuous if and only if, 
$\lim_{n\to\infty}\| F(\Delta_i)\|=0$ whenever $\Delta_i\downarrow\emptyset$.
\end{proposition}
\begin{proof}
Suppose that $\lim_{n\to\infty}\| F(\Delta_i)\|=0$ whenever $\Delta_i\downarrow\emptyset$. Let $\{\Delta_i\}_{i\in\mathbb{N}}$ be a disjoint sequence of sets such that $\cup_{i=1}^{\infty}\Delta_i=\Delta$. We have $\Delta-\cup_{i=1}^{n}\Delta_i\downarrow\emptyset$. Therefore,
\begin{equation*}
\lim_{n\to\infty}\|F(\Delta)-\sum_{i=1}^{n}F(\Delta_i)\|=\lim_{n\to\infty}\|F(\Delta-\cup_{i=1}^n\Delta_i)\|=0.
\end{equation*}
\noindent
Conversely, suppose $F$ uniformly continuous and $\Delta_i\downarrow\emptyset$. We can define the family of sets $\overline{\Delta}_i=\Delta_1-\Delta_i$. Clearly, $\overline{\Delta}_i\uparrow\Delta_1$.
Therefore, by proposition \ref{up},
$$\lim_{i\to\infty}\|F(\Delta_i)\|=\lim_{i\to\infty}\|F(\Delta_i)-F(\Delta_1)+F(\Delta_1)\|=\lim_{i\to\infty}\|F(\overline{\Delta}_i)-F(\Delta_1)\|=0.$$

\end{proof}

\noindent
Now, we can prove a necessary condition for the norm-1-property.

\begin{theorem}\label{uni}
Let $X$ be a Hausdorff topological space with a countable basis and locally homeomorphic to $\mathbb{R}^n$ for a particular $n$. (In other words, $X$ is an n-manifold. \cite{Munkres})  Suppose $F:X\to\mathcal{F(H)}$ is uniformly continuous and let $\sigma(F)$ be the spectrum of  $F$. Then, $F$ has the norm-1-property only if $\|F(\{x\})\|\neq 0$ for each $x\in\sigma(F)$.
\end{theorem}
\begin{proof}
We proceed by contradiction. Suppose that $F$ has the norm-1 property and that there exists $x\in\sigma(F)$, such that $\|F(\{x\})\|=0$. Let $\Delta_i\subset \mathcal{B}(\Omega)$ be a sequence of open sets such that $\Delta_i\downarrow\{x\}$. The existence of the sequence $\{\Delta_i\}_{i\in\mathbb{N}}$ is assured by the fact that $\Omega$ is locally homeomorphic to $\mathbb{R}^n$ for a particular $n\in\mathbb{N}$. By the norm-1 property\footnote{We recall (see \cite{Berberian}, page 32) that $x$ is in the spectrum of $F$ if and only if $x\in\Delta$ with $\Delta$ open implies $F(\Delta)\neq\mathbf{0}$.}, the uniform continuity of $F$, and proposition \ref{down},
\begin{eqnarray*}
1=\lim_{i\to\infty}\|F(\Delta_i)\|&=\lim_{i\to\infty}\|F(\Delta_i)-F(\{x\})+F(\{x\})\|\\
&\leq\lim_{i\to\infty}\|F(\Delta_i-\{x\})\|+\|F(\{x\})\|=0.
\end{eqnarray*}
\end{proof}
\noindent
Notice that absolute continuity with respect to a finite regular measure implies uniform continuity. 
\begin{theorem}\label{abs}
Let $F$ be absolutely continuous with respect to a finite measure $\nu$. Then, $F$ is uniformly continuous.
\end{theorem}
\begin{proof}
Suppose $\Delta_i\uparrow\Delta$. We have 
\begin{eqnarray*}
\lim_{n\to\infty}\| F(\Delta)-F(\Delta_i)\|&=\lim_{n\to\infty}\| F(\Delta-\Delta_i)\|\\
&\leq c\lim_{n\to\infty}\nu(\Delta-\Delta_i)=0.
\end{eqnarray*} 
\noindent
Proposition \ref{up} ends the proof.
\end{proof}
\noindent
In the case that $F$ is absolutely continuous with respect to an infinite measure, we have the following weak version of theorem \ref{abs}.
\begin{theorem}\label{abc}
Suppose $F$ absolutely continuous with respect to a regular measure $\nu$. Suppose, $\Delta$ such that $\nu(\Delta)<\infty$. Then, $F$ is uniformly continuous at $\Delta$; i.e., for any family of sets $\{\Delta_i\}_{i\in\mathbb{N}}$, $\Delta_i\downarrow\Delta$,
$$\lim_{n\to\infty}\|F(\Delta_i)-F(\Delta)\|=0.$$
\end{theorem}
\begin{proof}
Suppose $\Delta_i\downarrow\Delta$. Then, by the continuity of $\nu$, $\lim_{i\to\infty}\nu(\Delta_i)=\nu(\Delta)$. Hence, 
\begin{eqnarray*}
\lim_{n\to\infty}\| F(\Delta_i)-F(\Delta)\|&=\lim_{n\to\infty}\| F(\Delta_i-\Delta)\|\\
&\leq c\lim_{n\to\infty}\nu(\Delta_i-\Delta)=0.
\end{eqnarray*} 
\end{proof}
\noindent
As a consequence of theorem \ref{abs}, we have the following necessary condition for the norm-1 property in the case of absolutely continuous POVMs. The proof is analogous to the proof of theorem \ref{uni} and will be omitted.
\begin{theorem}\label{weak}
Let $F:\mathcal{B}(X)\to\mathcal{F(H)}$ be absolutely continuous with respect to a regular measure $\nu$. Then, $F$ has the norm-1 property only if $\|F(\{x\})\|\neq 0$ for each $x\in X$ such that $\nu(\{x\})<\infty$. 
\end{theorem}

\section{Localization in phase space}
In the present section, we show that localization observables in phase space cannot satisfy the norm-1 property. In the following the phase space is denoted by the symbol $\Gamma$. 
\noindent
First, we recall some key elements of the phase space approach to quantum mechanics \cite{Ali,Prugovecki,Schroeck1}. We follow References \cite{Schroeck1,Prugovecki}.

\noindent
The main idea is that one can represent the state $\rho$ of a quantum system (i.e., a trace class positive operator of trace $1$) by means of a distribution function $f^{\rho}(\mathbf{q},\mathbf{p})$ on a suitable phase space.  At variance with the Wigner approach \cite{Wigner}, the distribution functions are positive definite. Wigner's theorem \cite{Wigner} forbids that the marginals of the distribution functions satisfy the following relations
\begin{eqnarray}\label{marginal}
&\int_{\Delta}d\mathbf{q}\int f_\rho(\mathbf{q},\mathbf{p})\,d\mathbf{p}=\Tr(\rho\, Q(\Delta)),\\
&\int_{\Delta} d\mathbf{p}\int f_\rho(\mathbf{q},\mathbf{p})\,d\mathbf{q}=\Tr(\rho\, P(\Delta)),
\end{eqnarray}

where $Q(\Delta)$ and $P(\Delta)$ are the spectral measures corresponding to the position and momentum operators respectively. Relations (\ref{marginal}) are indeed replaced by 
\begin{eqnarray*}
&\int_{\Delta}d\mathbf{q}\int f_\rho(\mathbf{q},\mathbf{p})\,d\mathbf{p}=\Tr(\rho F^Q(\Delta)),\\
&\int_{\Delta} d\mathbf{p}\int f_\rho(\mathbf{q},\mathbf{p})\,d\mathbf{q}=\Tr(\rho F^P(\Delta)).
\end{eqnarray*}
\noindent
where $F^Q(\Delta)$ and $F^P(\Delta)$ are POVMs corresponding to $Q$ and $P$ respectively. In particular, $F^Q(\Delta)$ and $F^P(\Delta)$ are the smearing of the position and momentum operators $Q$ and $P$ 
\begin{eqnarray*}
F^Q(\Delta)&=\int_{\mathbb{R}}\omega_{\Delta}(x)\,dQ_x,\\
F^P(\Delta)&=\int_{\mathbb{R}}\nu_{\Delta}(x)\,dP_x,
\end{eqnarray*}
\noindent
and are called unsharp position and momentum observables \cite{B0,B1,B2,B3,B9}. The maps $\omega$ and $\nu$ are such that $\omega_{\Delta}(\cdot)$ and $\nu_{\Delta}(\cdot)$ are measurable functions for each $\Delta$ and $\omega_{(\cdot)}(x)$ and $\omega_{(\cdot)}(x)$ are probability measures for each $x$. They are usually called Markov kernels and describe a stochastic diffusion of the standard observables $Q$ and $P$ \cite{B6,B9}. That is why $F^Q$ and $F^P$ are usually called the unsharp version of the sharp observables $Q$ and $P$ respectively. All that shows that POVMs play a key role in the phase space formulation. Moreover, it is worth remarking that a derivation of classical and quantum mechanics in a unique mathematical framework is possible \cite{B10,B11}.

One of the main steps in this approach is the construction of the phase space. In brief, we can say that there is a procedure that starting from a Lie group $G$ allows the classification of all the closed subgroups $H\subset G$ such that $G/H$ is a simplectic space (i.e, a phase space).  For example, in the case of the Galilei group, a possible choice for $H$ is the group $H=SO_3$. Then, $\Gamma=G/H=\mathbb{R}^3\times\mathbb{R}^3$, which coincides with the phase space of classical mechanics. A different choice of $H$ generates a different phase space. In other words, the procedure allows the calculation of all the phase spaces corresponding to a locally compact Lie group, $G$, with a finite dimensional Lie algebra. Once we have the phase space, we can look for a strongly continuous unitary representation of $G$ in a Hilbert space $\mathcal{H}$ and then we can define the localization observable \cite{Schroeck1}.

\begin{definition}[See \cite{Schroeck1}.]\label{Local}
Let $G$ be a locally compact topological group, $H$ a closed subgroup of $G$, $U$ a strongly continuous unitary irreducible representation of $G$ in a complex Hilbert space $\mathcal{H}$ and $\mu$ a volume measure on $G/H$.    
A localization observable is represented by a  POVM 
$$A^{\eta}(\Delta)=\int_{\Delta}\vert U(\sigma(x))\eta\rangle\langle U(\sigma(x))\eta\vert\,d\mu(x).$$ 
\noindent
where, $\sigma:G/H\mapsto G$ is a measurable map and $\eta$ is a unit vector such that 
\begin{equation}\label{admissible}
\int_{G/H}\vert U(\sigma(x))\eta\rangle\langle U(\sigma(x))\eta\vert\,d\mu(x)=\mathbf{1}.
\end{equation}
\end{definition}
\begin{theorem}\cite{Schroeck1}
The POVM $A^{\eta}$ defined in Definition \ref{Local} is covariant with respect to $U$.
\end{theorem}
\noindent
A general property of the localization observables in Definition \ref{Local} is that they are absolutely continuous with respect to the measure $\mu$.
\begin{theorem}\cite{Schroeck1}\label{0}
The POVM in Definition \ref{Local} is absolutely continuous with respect to $\mu$.
\end{theorem}
\begin{proof}
For each $\psi\in\mathcal{H}$,
\begin{eqnarray*}
\langle \psi, A^{\eta}(\Delta)\psi\rangle&=\int_{\Delta}\langle\psi, U(\sigma(x))\eta\rangle\langle U(\sigma(x))\eta,\psi\rangle\,d\mu(x)\\
&=\int_{\Delta}\vert\langle\psi, U(\sigma(x)\eta\rangle\vert^2\,d\mu(x)\leq\int_{\Delta}d\mu(x).
\end{eqnarray*}
\end{proof}
\noindent
The localization of the photon in phase space was introduced in Ref. \cite{Brooke} with the same procedure we just described. Therefore, at variance with the usual definition of localization (where the covariance under the Euclidean group is required), localization in phase space requires that $F$ is covariant with respect to the group which describes the symmetry of the system (the Galilei group in the non-relativistic case and the Poincare group in the relativistic case). 

Before we prove that the norm-1 property is not possible for localization observables in phase space, we want to give a physical motivation which is based on the Heisenberg inequality. Let $F$ be a phase space localization observable covariant with respect to the Galilei group and suppose that the norm-1 property holds. In this case the phase space $\Gamma$ corresponding to the system can be chosen to be $\Gamma=\mathbb{R}_{\mathbf{q}}\times\mathbb{R}_{\mathbf{p}}=\mathbb{R}^3\times\mathbb{R}^3$ (See \cite{Schroeck1}, page 425.) Then, for each Borel set $\Delta_{\mathbf{q}}\times\Delta_{\mathbf{p}}\in\Gamma$, there exists a family of states $\psi_n$ such that 
\begin{equation}\label{1}
\lim_{n\to\infty}\langle\psi_n,F(\Delta_q\times\Delta_p)\psi_n\rangle=1
\end{equation}
\noindent
where, $\langle\psi_n,F(\Delta_q\times\Delta_p)\psi_n\rangle$ is interpreted as the probability that an outcome of a joint measurement of the unsharp position and momentum observables is in $\Delta_q\times\Delta_p$ when the state is $\psi_n$. Therefore, the  violation of Heisenberg inequality comes from the fact that (\ref{1}) holds for any Borel set $\Delta_q\times\Delta_p$.

In the following, we apply Theorem \ref{weak} to the case of the Galilei group $G=\{(t,\mathbf{q},\mathbf{p},R)\,\vert\, t\in\mathbb{R},\mathbf{q},\mathbf{p}\in\mathbb{R}^3, R\in SO(3)\}$ with $H=\{(t,\mathbf{0},\mathbf{0},R)\,\vert\, t\in\mathbb{R}, R\in SO(3)\}$. Therefore, $G/H=\mathbb{R}_{\mathbf{q}}\times\mathbb{R}_{\mathbf{p}}=\mathbb{R}^3\times\mathbb{R}^3$. In that case, the invariant measure is the Lebesgue measure. In the following we set $\mathbf{x}=(\mathbf{q},\mathbf{p})\in\mathbb{R}^3\times\mathbb{R}^3$.


\noindent
\noindent
Theorem \ref{weak} implies that the POVM $A^{\eta}$ in Definition \ref{Local} does not have the norm-1 property.
\begin{theorem}\label{nogo}
 The localization observable represented by the POVM $A^{\eta}$ with $G/H=\mathbb{R}^3\times\mathbb{R}^3$ and $\mu$ the Lebesgue measure does not have the norm-1 property.
\end{theorem}
\begin{proof} Let $\mathbf{x}\in G/H$. By Theorem \ref{0}, 
$$\|A^{\eta}(\{\mathbf{x}\})\|\leq\mu(\{\mathbf{x}\}).$$
\noindent
Since $\mu$ is the Lebesgue measure on $G/H$,
$$\|A^{\eta}(\{\mathbf{x}\})\|\leq\mu(\{\mathbf{x}\})=0.$$
\noindent
Theorem \ref{weak} completes the proof.
\end{proof}
\noindent
An analogous result can be proved in the case of massless relativistic particles. In the general relativistic case $G$ is the double cover $T^4\oslash SL(2,\mathbb{C})$ of the Poincare group. In particular, $T^4$ is the Minkowski space and $SL(2,\mathbb{C})$ is the double cover of the Lorentz group $L$, $A:SL(2,\mathbb{C})\to L$. The symbol $\oslash$ denotes the semidirect product.
In the massless relativistic case the relevant subgroup is $H=\mathbb{R}\oslash (\mathbb{R}^{2}\oslash \widetilde{O}(2))$ where $\widetilde{O}(2)$ is the double cover of the group of rotations in $\mathbb{R}^{2}$. 
The invariant measure on $G/H$ is (see equation (344), page 463, in Ref.\cite{Schroeck1})
\begin{equation}\label{Poincare}
d\mu=d(\alpha)d(\gamma)d(\delta)\times(p^0)^{-1}dp^1\wedge dp^2\wedge dp^3
\end{equation}
\noindent
where $\alpha =a_{\mu }(A[p_{0}])^{\mu },$ $\gamma =a_{\mu }(A[u_{0}])^{\mu
},$ $\delta =a_{\mu }(A[v_{0}])^{\mu },$ where $p_{0}=(1,0,0,1),$ $%
u_{0}=(0,1,0,0),$ $v_{0}=(0,0,1,0),$ and $(\boldsymbol{a},A)$ is an element
of the double cover of the Poincar\'{e} group. Thus $\alpha $, $\gamma $, $%
\delta $ are in $\mathbb{R}$. Hence, we have a representation of the zero
mass particles. 
Moreover, $\mu$ is zero in each single point subset of the phase space so that the reasoning in the proof of theorem \ref{nogo} can be used.  
\begin{theorem}\label{2}
If $G/H=T^4\oslash SL(2,\mathbb{C})/\mathbb{R}\oslash (\mathbb{R}^{2}\oslash \widetilde{O}(2))$ with the measure $\mu$  in equation (\ref{Poincare}), the POVM $A^{\eta}$ in Definition \ref{Local} does not have the norm-1 property.
\end{theorem}
\begin{proof} Let $\mathbf{x}\in G/H$. By Theorem \ref{0}, 
$$\|A^{\eta}(\{\mathbf{x}\})\|\leq\mu(\{\mathbf{x}\}).$$
\noindent
Since $\mu(\{\mathbf{x}\})=0$,
$$\|A^{\eta}(\{\mathbf{x}\})\|\leq\mu(\{\mathbf{x}\})=0.$$
\noindent
Theorem \ref{weak} completes the proof.
\end{proof}

\noindent


\section{Localization in Configuration Space}
\noindent
Now, we study the marginals of $A^{\eta}$ in the non-relativistic case and prove that they cannot have the norm-1 property. We limit ourselves to the marginal $F^Q_\eta(\Delta_{\mathbf{q}}):=A^{\eta}(\Delta_{\mathbf{q}}\times\mathbb{R}_{\mathbf{p}})$ which represents the unsharp position observable. Clearly what we prove applies also to the marginal $F^P_\eta(\Delta_{\mathbf{p}}):= A^{\eta}(\mathbb{R}_{\mathbf{q}}\times\Delta_{\mathbf{p}})$ which represents the unsharp momentum observable. 
\begin{theorem}\label{Q}
The POVM $F^Q_{\eta}(\Delta_{\mathbf{q}})$ is absolutely continuous with respect to the Lebesgue measure on $\mathbb{R}_{\mathbf{q}}$.
\end{theorem} 
\begin{proof}
\begin{eqnarray*}\label{position}
F^Q_\eta(\Delta_{\mathbf{q}})=A^{\eta}(\Delta\times\mathbb{R}_{\mathbf{p}})&=\int_{\Delta\times\mathbb{R}_{\mathbf{p}}}U(\sigma(\mathbf{q},\mathbf{p}))\,\vert\eta\rangle\langle\eta\vert\, U^\dag(\sigma(\mathbf{q},\mathbf{p}))\,d\mathbf{q}\,d\mathbf{p}\\
&=\int_{\Delta}\,d\mathbf{q}\int_{\mathbb{R}_{\mathbf{p}}} U(\sigma(\mathbf{q},\mathbf{p}))\,\vert\eta\rangle\langle\eta\vert\, U^\dag(\sigma(\mathbf{q},\mathbf{p}))\,d \mathbf{p}\\
&=\int_{\Delta}\widehat{Q}_\eta(\mathbf{q})\,d\mathbf{q}\leq\int_{\Delta}\mathbf{1}\,d\mathbf{q},
\end{eqnarray*}
\noindent
where 
$$\widehat{Q}_\eta(\mathbf{q})=\int_{\mathbb{R}_{\mathbf{p}}}U(\sigma(\mathbf{q},\mathbf{p}))\,\vert\eta\rangle\langle\eta\vert\, U^\dag(\sigma(\mathbf{q},\mathbf{p}))\,d\mathbf{p}$$
\noindent
and equation (\ref{admissible}) in definition \ref{Local} has been used.\end{proof}

\noindent
Theorems \ref{weak} implies the following corollary.  
\begin{corollary} $F^Q_{\eta}$ cannot have the norm-1 property
\end{corollary}

In Ref. \cite{Busch1} it is shown that in order for a localization observable to satisfy Einstein causality, the localization observable must be commutative. It is worth remarking that, although $A^\eta$ is not commutative, $F^Q_\eta$ is commutative and can be characterized as the smearing of the position operator  \cite{B0}-\cite{B9}. Unfortunately, as we have just proved, $F^Q_{\eta}$ does not satisfy the norm-1 property. It would be interesting to analyze in general the relationships between causality and norm-1 property. That will be the topic of a future work.  


\section*{References}
\begin{thebibliography}{40}
\bibitem{Berberian} S.K. Berberian, "Notes on Spectral Theory," D. Van Nostrand Company, Inc., Princeton, New Jersey (1966). 
\bibitem{Ali} S.T. Ali, H. D. Doebner, J. Math. Phys. Vol. 17,  1976, pp. 1105-1976.\bibitem{Hegerfeldt}. G.C. Hegerfeldt, "Remarks on Causality and Particle Localization," Phys. Rev. D Vol. 10, 1974, pp. 3320-3321.
\bibitem{Castrigiano} D.P.I. Castrigiano, "On Euclidean systems of covariance for massless particles," Letters in Mathematical Physics, Vol. 5 1981, pp. 303-309.
\bibitem{Jauch} J.M. Jauch, C. Piron, Helv. Phys. Acta, Vol. 40, 1967, pp. 559.
\bibitem{Rosewarne} D. Rosewarne, S. Sarkar, "Rigorous theory of photon localizability," Quant. Optics, Vol. 4, 1992, pp. 405-413.
\bibitem{Busch} P. Busch, M. Grabowski, P. Lahti: "Operational quantum physics," Lecture Notes in Physics, Vol. 31, Springer-Verlag, Berlin, (1995). 
\bibitem{Busch1} P. Busch, "Unsharp Localization and Causality in Relativistic Quantum Theory," J.Phys. A: Math. Gen., Vol. 32, 1999, pp. 6535.

\bibitem{Holevo1} A. S. Holevo, "Probabilistic and statistical aspects of quantum theory," North Holland, Amsterdam, 1982.
\bibitem{Holevo2} A.S. Holevo, Lecture Notes in Physics, Vol. 1055, 1984, pp. 153-172.\bibitem{Muga} J.G. Muga, R. Sala Mayato, I. L. Egusquiza (Eds.) "Time in Quantum Mechanics," Lecture Notes in Physics, Vol. 72, Springer, Berlin Heidelberg (2008).
\bibitem{Schroeck1} F. E. Schroeck, Jr., "Quantum Mechanics on Phase Space," Kluwer Academic Publishers, Dordrecht, 1996.
\bibitem{Brooke} J.A. Brooke, F.E. Schroeck, "Localization of the photon on phase space," Vol. 37, 1996, pp. 5958. 
\bibitem{Schroeck} F. E. Schroeck Jr., Int. J. Theor. Phys., Vol. 44, (11), 2005, pp. 2091-2100.
\bibitem{Newton} T. D. Newton and E. P. Wigner, Revs. Modern Phys. Vol. 21 (1949) pp.400.
\bibitem{Wightman} A.S. Wightman, Rev. Mod. Phys., Vol. 34, (4), 1962. 
\bibitem{Amrein} W.O. Amrein, Helvetica Physica Acta, Vol. 42, 1969, pp. 149-190.
\bibitem{Varadarajan} V.S. Varadaraian, "Geometry of quantum theory," Springer Verlag, New York-Berlin Heidelberg-Tokyo 1985.

\bibitem{Car} C. Carmeli, T. Heinonen, A. Toigo, J. Math. Phys. Vol. 45, 2004, pp. 2526-2539.


\bibitem{He} T. Heinoinen, P. Lahti, J. P. Pellopp\accent4 a\accent4 a, S. Pulmannov\'a, K. Ylinen, J. Math. Phys., Vol. 44,  2003, pp. 1998-2008.
\bibitem{Ludwig} G. Ludwig, \textit{Foundations of Quantum Mechanics I}, (Springer-Verlag, New York, Heidelberg, Berlin, 1983). 
\bibitem{Muynck} W.M. de Muynck, "Foundations of Quantum Mechanics, an Empiricist Approach", Kluwer Academic Publishers, Dordrecht, Boston, London (2002).
\bibitem{Munkres} J.R. Munkres, "Topology," Prentice Hall, Inc., Upper Saddle River, NJ (2000). 
\bibitem{Prugovecki} E. Prugove\v{c}ki: "Stochastic Quantum Mechanics and Quantum Spacetime," D. Reidel Publishing Company, Dordrecht, Holland, 1984.
\bibitem{Wigner} E. P. Wigner, "Quantum Mechanical Distribution Functions Revisited," in Perspectives in Quantum Theory, W. Yourgrau and A. van der Merwe (eds.), 25, MIT Press, Cambridge, Mass 1971.
\bibitem{B0} R. Beneduci, G. Nistic\'o, J. Math. Phys., Vol. 44,  2003, pp. 5461.
\bibitem{B1} R. Beneduci, J. Math. Phys., Vol. 47, 2006, pp. 062104.
\bibitem{B2} R. Beneduci, Int. J. Geom. Meth. Mod. Phys. Vol. 3, 2006, pp. 1559.
\bibitem{B3} R. Beneduci, J. Math. Phys., Vol. 48, 2007, pp. 022102.
\bibitem{B4} R. Beneduci, Il Nuovo Cimento B, Vol. 123, 2008, pp. 43-62.
\bibitem{B5} R. Beneduci, Int. J. Theor. Phys., Vol. 49, 2010, pp. 3030-3038.
\bibitem{B6} R. Beneduci, Bull.  Lond. Math. Soc., Vol. 42,  2010, pp. 441-451.
\bibitem{B7} R. Beneduci, Linear Algebra and its Applications, Vol. 43, 2010, pp. 1224-1239.
\bibitem{B8} R. Beneduci, "On the relationships between the moments of a POVM and the generator of the von Neumann algebra it generates," International Journal of Theoretical Physics, Vol. 50, 2011, pp. 3724-3736, doi: 10.1007/s10773-011-0907-7.
\bibitem{B9} R. Beneduci, "Semispectral Measures and Feller Markov Kernels", arXiv:1207.0086v1,  30 June 2012.
\bibitem{B10} R. Beneduci, J. Brooke, R. Curran, F. Schroeck Jr., "Classical Mechanics in Hilbert Space, part I," International Journal of Theoretical Physics, Vol. 50, 2011, pp. 3682-3696, doi: 10.1007/s10773-011-0797-8.
\bibitem{B11} R. Beneduci, J. Brooke, R. Curran, F. Schroeck Jr. "Classical Mechanics in Hilbert Space, part II," International Journal of Theoretical Physics, Vol. 50, 2011, pp. 3697-372,3  doi: 10.1007/s10773-011-0869-9.
\end{thebibliography}
\end{document}